\documentclass[12pt]{article}

\usepackage{etex}
\usepackage{amssymb,amsthm,amscd,amsbsy,array}
\usepackage{color}
\usepackage[colorlinks=true, pdfstartview=FitV, linkcolor=blue, citecolor=blue, urlcolor=blue]{hyperref}

\usepackage{tikz}
\usetikzlibrary{calc,through,backgrounds}
\usetikzlibrary{arrows,decorations.pathmorphing,positioning,fit,petri}

\let\ssection=\section
\renewcommand{\section}{\setcounter{equation}{0}\ssection}

\newcommand{\AdS}{\mathrm{AdS}}

\newcommand{\Aut}{\mathrm{Aut}}

\newcommand{\cB}{{\mathcal{B}}}

\newcommand{\bone}{\boldsymbol{1}}

\newcommand{\bbC}{\mathbb{C}}

\newcommand{\Conf}{\mathrm{Conf}}
\newcommand{\conf}{\mathrm{conf}}

\newcommand{\Diff}{{\mathrm{Diff}}}
\newcommand{\diag}{\mathrm{diag}}

\newcommand{\Div}{\mathrm{Div}}

\newcommand{\End}{{\mathrm{End}}}
\newcommand{\rE}{{\mathrm{E}}}
\newcommand{\cE}{{\mathcal{E}}}
\newcommand{\Ein}{{\mathrm{Ein}}}

\newcommand{\cF}{{\mathcal{F}}}

\newcommand{\rg}{\mathrm{g}}
\newcommand{\hrg}{\widehat{\rg}}
\newcommand{\trg}{\widetilde{\rg}}

\newcommand{\rh}{\mathrm{h}}

\newcommand{\Id}{\mathrm{Id}}
\newcommand{\Int}{\mathrm{Int}}

\newcommand{\sfL}{\mathsf{L}}

\newcommand{\cM}{\mathcal{M}}

\newcommand{\hM}{\widehat{M}}

\newcommand{\rO}{{\mathrm{O}}}
\newcommand{\ro}{{\mathrm{o}}}
\newcommand{\cO}{{\mathcal{O}}}

\newcommand{\bbP}{\mathbb{P}}
\newcommand{\oP}{{\overline{P}}}

\newcommand{\cQ}{\mathcal{Q}}
\newcommand{\oQ}{{\overline{Q}}}

\newcommand{\hr}{\widehat{r}}

\newcommand{\bbR}{\mathbb{R}}

\newcommand{\Ric}{\mathrm{Ric}}

\newcommand{\hs}{\widehat{s}}

\newcommand{\Sch}{\mathrm{Sch}}
\newcommand{\sch}{\mathrm{sch}}

\newcommand{\SL}{\mathrm{SL}}
\newcommand{\Sl}{\mathrm{sl}}

\newcommand{\rso}{\mathrm{so}}

\newcommand{\so}{\mathrm{o}}

\newcommand{\rT}{\mathrm{T}}

\newcommand{\Tr}{\mathrm{Tr}}
\newcommand{\wht}{{\widehat{t}}}
\newcommand{\htheta}{{\widehat{\theta}}}

\newcommand{\rU}{\mathrm{U}}

\newcommand{\Vect}{\mathrm{Vect}}

\newcommand{\Vol}{\mathrm{Vol}}

\newcommand{\hx}{{\widehat{x}}}

\newcommand{\oX}{{\overline{X}}}

\newcommand{\hxi}{\widehat{\xi}}
\newcommand{\txi}{\widetilde{\xi}}

\newcommand{\oY}{{\overline{Y}}}

\newcommand{\bbZ}{\mathbb{Z}}

\newcommand{\half}{\frac{1}{2}}

\newcommand{\medbox}[1]{\fbox{%
\rule[-10pt]{0pt}{25pt}$\;\;\displaystyle{#1}\;\;$}%
}

\begin{document}

\baselineskip=24pt

\oddsidemargin .1truein
\newtheorem{thm}{Theorem}[section]
\newtheorem{lem}[thm]{Lemma}
\newtheorem{cor}[thm]{Corollary}
\newtheorem{pro}[thm]{Proposition}
\newtheorem{ex}[thm]{Example}
\newtheorem{rmk}[thm]{Remark}
\newtheorem{defi}[thm]{Definition}

\title{\sc Schr\"odinger manifolds}

\author{
C. DUVAL\footnote{mailto: duval-at-cpt.univ-mrs.fr}
\quad\hbox{and}\quad
S. LAZZARINI\footnote{mailto: lazzarini-at-cpt.univ-mrs.fr}\\[5mm]
Centre de Physique Th\'eorique,\footnote{ 
Laboratoire affili\'e \`a la FRUMAM
}\\ [2mm]
Aix-Marseille Univ, CNRS UMR-7332,  Univ Sud Toulon-Var,\\[2mm]
F--13288 Marseille Cedex 9, France\\[2mm]
}


\date{August 13, 2012}

\maketitle

\thispagestyle{empty}

\begin{abstract}
This article propounds, in the wake of influential work of Fefferman and Graham about Poincar\'e extensions of conformal structures, a definition of a (Poincar\'e-)Schr\"odinger
manifold whose boundary is endowed with a conformal Bargmann structure above a non-relativistic Newton-Cartan spacetime.  Examples of such manifolds are worked out in terms of homo\-geneous spaces of the Schr\"odinger group in any spatial dimension, and their global topo\-logy is carefully analyzed.  These archetypes of Schr\"odinger manifolds carry a Lorentz structure together with a prefer\-red null Killing vector field; they are shown to admit the Schr\"odinger group as their maximal group  of isometries. The relationship to similar objects arising in the non-relativistic AdS/CFT correspondence is discussed and clarified.
\end{abstract}








\newpage

\tableofcontents

\newpage

\section{Introduction}
\label{IntroSection}

The notion of ``non-relativistic conformal symmetry'' goes back to Jacobi and Lie as highlighted in, e.g., \cite{RU,DH}. In the early seventies, Jackiw \cite{Jack}, Niederer \cite{Nie}, and Hagen \cite{Hag} rediscovered this symmetry within the quantum mecha\-nical context; the maximal kinematical symmetry group of the free Schr\"odinger equation has since then been coined the ``Schr\"odinger group''. One of the key features of this sym\-metry is a specific action of dilations according to which time is dilated twice as much as space (the \textit{dynamical exponent} is $z=2$). A remarkable relationship between the Schr\"odinger Lie algebra and the relativistic conformal Lie algebra (in suitable dimensions) was then  unveiled in \cite{BPS}. The Schr\"odinger symmetry also happened to play a central r\^ole in the physics of strongly anisotropic critical systems \cite{Hen}, and in the description of ageing phenomena \cite{Hen2,Hen3}. At a geo\-metrical level, the (center-free) Schr\"odinger group has been interpreted as the group of those ``conformal transformations'' of a Newton-Cartan spacetime that also permute its unparametrized geodesics~\cite{Duv}. It soon became patent  that an adapted framework to deal intrinsically  with non-relativistic con\-formal symmetries is provided by Bargmann structures \cite{DBKP} defined on $(\bbR,+)$- or circle-bundles over Newton-Cartan spacetimes. Let us recall that a Bargmann manifold (akin to generalized \textit{pp}-waves \cite{EK}) is such a principal fibre-bundle, endowed with a Lorentz metric, whose fundamental vector field is null and covariantly constant. This definition entailed that the conformal automorphisms of a Bargmann structure constitute a Lie group which turns out to be actually isomorphic with the Schr\"odinger group, yielding henceforth a clear-cut geometrical status to the latter \cite{DGH,DHH}; see also \cite{PBD}. We refer to  \cite{RU} for a modern and recent review of the Schr\"odinger and Schr\"odinger-Virasoro symmetries.

\goodbreak

A few years ago, Son \cite{Son} and, independently, Balasubramanian and McGreevy~\cite{Bala} have put forward a geometrical realization of the Schr\"odinger group as a group of \textit{iso\-metries} of some Lorentz metric on a two-parameter spacetime extension, in the framework of the AdS/CFT cor\-respondence initiated by Maldacena \cite{Mal,AGMOO}. 
From then on, this \textit{non-relativistic holography}
has triggered much interest within a wide range of subjects, for instance, in ageing-gravity duality \cite{MP}, non-relativistic field theory \cite{Gol}, string theory and black hole
physics \cite{MMT,LW}. In condensed matter physics, the Schr\"odinger group turns out to be
also the  dynamical symmetry group of the two-body interactions in ultracold
fermionic atoms \cite{Bala,Son}. See also recent work \cite{BMM} on the conformal symmetries of the unitary Fermi gas.

It has been pointed out in \cite{RS,MMT,BHR,CBDY,HR} that non-relativistic conformal sym\-metries for backgrounds arising in string theory and black hole geo\-metry should be best viewed as asymptotic symmetries of AdS spacetime associated with a definite notion of conformal boundary. These authors also called attention to work of Fefferman-Graham \cite{FG} which should, conversely, provide efficient geo\-metrical means to deal with expansions of these asymptotic symmetries to the bulk spacetime, endowed with a \textit{Poincar\'e metric}. This is precisely the viewpoint we will espouse in this article in an effort to adapt the Fefferman-Graham (FG) construct to the particular instance of the Schr\"odinger symmetry. We mention, in this vein, another approach using the alternative notion of \textit{ambient space} \cite{FG} used in \cite{LN} to describe conformal {\em pp}-waves.

Let us now recall that the above-mentioned extension of the AdS/CFT correspondence to
non-relativistic field theory is based on using the (locally defined) metric \cite{Bala,Son,BHR,LW}
\begin{equation}
\hrg=
\frac{1}{r^2}\left[
\sum_{i=1}^d{(dx^i)^2}+2dtds+dr^2-\frac{dt^2}{r^2}\right]
\label{BSmetric}
\end{equation}
on a $(d+3)$-dimensional relativistic spacetime, whose
key property is that its group of {\it iso\-metries} is the \textit{Schr\"odinger group} of non-relativistic {\it conformal}
transformations 
of $(d+1)$-dimensional
Galilei spacetime, coordinatized by $(x^1,\ldots,x^d,t)$.

The purpose of this article is to provide an appropriate geometrical interpretation of such a manifold which, as we will show, turns out to be an instance of what we will call a ``Schr\"odinger manifold''. Let us underline that we will consider the only cases where the spatial dimension is $d>0$. (See, e.g., \cite{HM-T} for a thorough study of the AdS/CFT correspondence in the case $d=0$.)

We now summarize the main outcome and results of this article. 

Our approach strongly relies, on the one hand, on the general notion of a conformal Bargmann structure above non-relativistic spacetime (Definition \ref{ConfBargStrDef}), and, on the other hand, on an adaptation to this non-relativistic conformal structure of the FG formal theory of Poincar\'e metrics. This standpoint will help us introduce, via Definition \ref{SchManifDef}, the novel notion of \textit{Schr\"odinger manifold}, endowed with both a Poincar\'e metric and a null Killing vector field, and whose conformal boundary corresponds precisely to our original conformal Bargmann structure. Such Schr\"odinger manifolds are, indeed, exemplified by the Poincar\'e metric $\rg^+=\hrg+dt^2/r^4$ and the null Killing vector field $\partial/\partial{s}$, read off Equation (\ref{BSmetric}). This is the content of Theorem \ref{mainThm}, the main upshot of our article. We will furthermore prove (Proposition \ref{ProHomSpace}) that this emblematic example actually stems from a certain homogeneous space, $\hM$, of the Schr\"odinger group $\Sch(d+1,1)$, the latter being the maximal group of isometries of $\hM$ (Proposition \ref{mainThm0}).

The article is organized as follows.

Section \ref{ConfBargSection} introduces the basics of conformal non-relativistic geometry, namely the definition of a \textit{conformal Bargmann structure} above a Newton-Cartan structure on ($d+1$)-dimensional spacetime. We recall, and put in a geometrical guise, the covariant Schr\"odinger equation and its relationship with the conformal Laplace (Yamabe) operator acting on densities. The Schr\"odinger group is then naturally intro\-duced in terms of the automorphisms of a conformal Bargmann structure.

\goodbreak

Our definition of \textit{Schr\"odinger manifolds} is presented in Section \ref{SchManifSection}. It funda\-mentally relies on the construction of the ``Poincar\'e'' formal defor\-mation of a conformal (pseudo-)Riemannian structure due to Fefferman and Graham. Emphasis will be put on the Lorentzian case, relevant to deform conformal Bargmann structures. The r\^ole of a special null Killing vector field will also be highlighted in the definition of Schr\"odinger manifolds. 

Section \ref{SchLieGroupSection} is concerned with the global structure of the Schr\"odinger group. The Schr\"odinger Lie algebra, $\sch(d+1,1)$, spanned by the vector fields (\ref{schd+1,1}) of the flat Bargmann structure is chosen to be integrated inside the conformal group $\rO(d+2,2)$ viewed as the group of isometries of the ambient vector space $\bbR^{d+4}$, endowed with the metric (\ref{G1}). The \textit{Schr\"odinger group} $\Sch(d+1,1)$ is then defined as the stabilizer of some nilpotent element, $Z_0$, in the Lie algebra $\ro(d+2,2)$. Its relationship with the manifold of null geodesics of compactified Minkowski space $\Ein_{d+1,1}$ is revealed.

Section \ref{HomSchManifSection} gathers the main results of the article. Much in the spirit of the Klein program, we seek examples of Schr\"odinger manifolds as homogeneous spaces of the Schr\"odinger group $\Sch(d+1,1)$ itself. The outcome is given by Propositions~\ref{ProHomSpace} and \ref{ProConfBargHomSpace}. These \textit{homogeneous Schr\"odinger manifolds} turn out to be open submanifolds of $\AdS_{d+3}$, and their topology is completely worked out (see  Figure \ref{Fig1} for an illustration). We furthermore show that the Schr\"odinger group is actually their maximal group of isometries. The AdS/CFT metric (\ref{BSmetric}) acquires, hence, a global status as the canonical metric of our Schr\"odinger-homogeneous space arising as a (Poincar\'e-)Schr\"odinger metric inherited from the FG construction.

In Section \ref{ConclusionSection} we summarize the content of the article and draw several conclusions. We also offer perspectives related, among others, to open problems regarding the existence and uniqueness of Schr\"odinger manifolds.


\section{Schr\"odinger equation and conformal Bargmann structures}\label{ConfBargSection}

Let us recall that a \textit{Bargmann structure} \cite{DBKP} is a principal $H$-bundle $\pi:M\to\cM$ over a $(d+1)$-dimensional smooth manifold $\cM$, where $H\cong(\bbR,+)$ or $\rU(1)$; its total space, $M$, is assumed to carry a Lorentz metric, $\rg$, the
fundamental vector field, $\xi$, of the $H$-action being null, $\rg(\xi,\xi)=0$, and covariantly constant  with respect to the Levi-Civita connection, $\nabla\xi=0$. 

It has been proved \cite{DBKP} that a Bargmann structure $(M,\rg,\xi)$ projects onto a ``Newton-Cartan'' (NC) structure on non-relativistic spacetime $\cM=M/H$ \cite{Kunzle}.  The nowhere vanishing $1$-form $\theta=\rg(\xi)$ associated with $\xi$ via the metric, $\rg$, is closed; it therefore descends onto the \textit{time} axis $T=M/\ker(\theta)$ as a $1$-form which we call the ``clock'' of the structure. Bargmann structures are interpreted as \textit{generalized pp-waves} in general relativity; see, e.g., \cite{EK,KSMCH,DGH}.  

We recall that the canonical \textit{flat} Bargmann structure on~$M=\bbR^{d+2}$, with $H=(\bbR,+)$, is given by
\begin{equation}
\medbox{
\rg=\sum_{i=1}^d dx^i\otimes{}dx^i+2dt\odot{}ds
\qquad
\&
\qquad
\xi=\frac{\partial}{\partial s}
}
\label{FlatBargmannStructure}
\end{equation}
where we have put $t=x^{d+1}$, and $s=x^{d+2}$; we will use the shorthand notation $\bbR^{d+1,1}=(\bbR^{d+2},\rg)$; also will ``$\odot$'' denote the symmetrized tensor product. Here $(x^1,\ldots,x^d)$ are ``spatial'' coordinates, and $t$ stands for the absolute time co\-ordinate on Galilei spacetime such that
\begin{equation}
\theta=dt
\label{theta}
\end{equation}
while $s$ is a 
coordinate homo\-geneous to an action per mass. 

\subsection{Covariant Schr\"odinger equation}

We first introduce the useful notion of $\lambda$-densities spanning the $\Diff(M)$-module $\cF_\lambda(M)$ whose elements can be locally written as $\Psi=f\vert\Vol\vert^\lambda$, with $f\in{}C^\infty(M,\bbC)$, if $\Vol$ is a volume element of $M$. The associated $\Vect(M)$-module structure of~$\cF_\lambda(M)$ is then defined via the Lie derivative 
$\sfL_X^{\!\lambda}{f}=X(f)+\lambda\Div(X)f$, for all $X\in\Vect(M)$.

Let $(M,\rg)$ be a $n$-dimensional pseudo-Riemannian manifold. We recall that the \textit{Yamabe operator}, or conformal Laplacian \cite{Bes,DO,East}, is the conformally-invariant dif\-ferential operator 
$\Delta_\rg^\conf:\cF_\frac{n-2}{2n}(M)\to\cF_\frac{n+2}{2n}(M)$ defined by
$\Delta_\rg^\conf=\Delta_\rg-\frac{n-2}{4(n-1)}R(\rg)$, where~$R(\rg)$ denotes the scalar curvature of the Levi-Civita connection of $(M,\rg)$. 

\begin{pro}\label{SchEqPro}
Given a Bargmann manifold $(M,\rg,\xi)$ of dimension $n=d+2$, the system 
\begin{equation}
\Delta_\rg^\conf\,\Psi=0
\qquad
\&
\qquad
\frac{\hbar}{i}\sfL_\xi^{\!\lambda}\,\Psi=m\,\Psi
\label{SchEq}
\end{equation}
with $\lambda=\frac{d}{2d+4}$ descends as the covariant Schr\"odinger equation of mass $m$ on the associated Newton-Cartan spacetime.
\end{pro}
The proof of Proposition \ref{SchEqPro} relies essentially on the derivation given in \cite{DBKP}, and on the fact that the fundamental vector field, $\xi$, is divergencefree, $\Div(\xi)=0$. In the latter reference, the NC field equations, $\Ric(\rg)=4\pi{}G\varrho\,\theta\otimes\theta$, where~$\varrho$ stands for the mass density of the sources, were assumed to hold, thus implying $R(\rg)=0$.

\begin{rmk}
{\rm
The structural group $H$ may be compact in some special
instances, e.g., $H=\rU(1)$ for a Taub-NUT like solution of NC field
equations. This leads, in view of the second equation in
(\ref{SchEq}), to the quantization of mass \cite{DGH}. From now on, we
shall be mainly concerned with the case $H=(\bbR,+)$. 
}
\end{rmk}

\subsection{Symmetries of the Schr\"odinger equation}

Denote by $[\Phi\mapsto{}\Phi_\lambda]$ the action of $\Diff(M)$ on $\cF_\lambda(M)$. 
A \textit{symmetry} of the Schr\"odinger equation is a local diffeomorphism $\Phi\in\Diff_\mathrm{loc}(M)$ such that
\begin{equation}
\Delta_\rg^\conf\circ{}\Phi_\lambda=\Phi_\mu\circ\Delta_\rg^\conf
\qquad
\&
\qquad
\sfL_\xi^{\!\lambda}\circ{}\Phi_\lambda =\Phi_\lambda\circ{}\sfL_\xi^{\!\lambda}
\label{SymmSchEq}
\end{equation}
with the weights $\lambda=\frac{d}{2d+4}$, and $\mu=\frac{d+4}{2d+4}$.

\begin{pro}\cite{DGH}
\label{SymSchPro}
The symmetries of the Schr\"odinger equation 
form the ``Schr\"o\-dinger (pseudo-)group''
$
{
\Sch(M,\rg,\xi)=\Conf_\mathrm{loc}(M,\rg)\cap\Aut(M,\xi)
}
$
consisting of those $\Phi\in\Diff_\mathrm{loc}(M)$ such that
\begin{equation}
\medbox{
\Phi^*\rg=\Omega^2_\Phi\cdot\rg
\qquad
\&
\qquad
\Phi_*\xi=\xi
}
\label{SchroedingerAutomorphism}
\end{equation}
for some $\Omega_\Phi\in{}C^\infty(M,\bbR^*_+)$ depending on $\Phi$.\footnote{We will confine considerations to conformal dif\-feomorphisms of $(M,\rg)$ that commute with the $H$-action on $M$, hence satisfying (\ref{SchroedingerAutomorphism}). 
We will \textit{not} consider, here, the larger (pseudo-)group of all conformal transformations $\Phi$ of $(M,\rg)$ that permute the $H$-orbits, i.e., such that $\xi\wedge\Phi_*\xi=0$.} These $\Phi\in\Sch(M,\rg,\xi)$ permute, hence, the solutions of the Schr\"odinger equation (\ref{SchEq}) ac\-cording to
\begin{equation}
\Psi\mapsto(\Phi_\lambda)_*\Psi
\label{SchAction}
\end{equation}
with $\lambda=\frac{d}{2d+4}$.
This group descends onto NC spacetime $\cM$ as the ``center-free Schr\"odinger (pseudo-)group'' $\Sch(M,\rg,\xi)/H$.
\end{pro}

\subsection{Conformal Bargmann structures}\label{ConfBargStrSubSection}

In view of Proposition \ref{SymSchPro}, one of the fundamental geometrical objects associated with the Schr\"odinger equation is clearly the conformal class $[\rg]$ of the Bargmann metric $\rg$ on (extended spacetime) $M$. 
Indeed, given any $\overline{\rg}\in[\rg]$, one duly has $\rg(\xi,\xi)=\overline{\rg}(\xi,\xi)=0$; now, to further insure $\nabla\xi=\overline{\nabla}\xi=0$, i.e., that $\rg$ and $\overline{\rg}$ are Bargmann-equivalent, one finds
\begin{equation}
\rg\sim\overline{\rg}
\qquad
\iff
\qquad
\overline{\rg}=\Omega^2\rg
\quad
\&
\quad
d\Omega\wedge\theta=0
\label{EqRel}
\end{equation}
which infers that the conformal factor, $\Omega$, be a function of the time axis,~$T$.

\goodbreak

\begin{defi}\label{ConfBargStrDef}
A ``conformal Bargmann structure'' is an equivalence class $(M,[\rg],\xi)$ of Bargmann manifolds for the equivalence relation (\ref{EqRel}).
\end{defi}
Such a structure basically involves a \textit{conformal class} of Lorentz metrics on a principal fibre bundle $\pi:M\to\cM$ with structure group $(\bbR,+)$ --- or $\rU(1)$ --- whose fundamental vector field is lightlike, and parallel.

The Schr\"odinger group, as defined in Proposition \ref{SymSchPro}, is therefore isomorphic to the group of automorphisms of the conformal Bargmann structure defined in Proposition \ref{ConfBargStrDef}, namely $\Sch(M,\rg,\xi)\cong\Aut(M,[\rg],\xi)$.


\section{General definition of Schr\"odinger manifolds}\label{SchManifSection}

We have, so far, unveiled new geometrical structures involving
conformal structures in the presence of a null, parallel, and nowhere
vanishing vector field admitting a clear-cut physical interpretation via the definition of the mass in the Schr\"odinger equation
(\ref{SchEq}). 
Accordingly, our main goal will now be to specialize the
FG definition of ``Poincar\'e metrics'' associated with
conformal structures to our particular, non-relativistic, framework
featuring conformal Bargmann structures. 

\subsection{Formal theory of Poincar\'e metrics \& conformal infinity ac\-cording to Fefferman-Graham}\label{sectionFG}

In their quest of \textit{conformal invariants} of a conformal structure of signature $(p,q$), Fefferman and Graham \cite{FG} have devised two equivalent constructs:
\begin{enumerate}
\item
the {ambient metric} on a pseudo-Riemannian manifold of signature $(p+1,q+1)$,
\item\label{Poincare}
the {Poincar\'e metric} on a pseudo-Riemannian manifold of signature $(p+1,q)$.
\end{enumerate}

\goodbreak

In order to make contact with the aforementioned physics literature, involving local expressions for the metric in Poincar\'e patches, we restrict further considera\-tions to item \ref{Poincare}. Let us hence recall the general definition \cite{FG} of Poincar\'e metrics that will appear as the cornerstone of the subsequent study.

Start with a manifold $M$ and a conformal class $[\rg]$ of metrics of signature $(p,q)$, such that $n=p+q>2$. Consider now a manifold ${M^+}$ such that $M=\partial{M^+}$. Let $r\in{}C^\infty({M^+})$ verify $r>0$ in $\Int({M^+})$, and ${r=0}$ \& ${dr\neq0}$ on $\partial{M^+}$ (this function is called a \textit{defining function} for $M$). A metric ${\rg^+}$ of signature $(p+1,q)$ on $\Int(M^+)$ is ``conformally compact'' iff $r^2\,{\rg^+}$ extends smoothly to ${M^+}$ and $r^2\,{\rg^+}\strut\vert_{TM}$ is non-degenerate \cite{Pen}.

\begin{defi}\label{FGconfInf}\cite{FG}
We say that $({M^+},{\rg^+})$ has $(M,[\rg])$ as {conformal infinity} whenever 
$
r^2\,{\rg^+}\strut\vert_{TM}\in[\rg].
$
\end{defi}

\begin{defi}\label{FG}\cite{FG}
A {Poincar\'e metric} for $(M,[\rg])$ is a pair $({M^+},{\rg^+})$ where $M^+$ is an open neighborhood of $M\times\{0\}$ in $M\times\bbR^+$ such that
\begin{itemize}
\item
$(M^+,\rg^+)$ has $(M,[\rg])$ as conformal infinity
\item
$(M^+,\rg^+)$ is an asymptotic solution of Einstein's equation
$\Ric({\rg^+})+k{\rg^+}=0$ (normalization condition: $k=n$).
\end{itemize}
\end{defi}

Poincar\'e metrics admit the local expression
\begin{equation}{\rg^+}=\frac{1}{r^2}\left[\sum_{i,j=1}^n{\rg^+_{ij}(x,r)dx^i\otimes{}dx^j+dr\otimes{}dr}\right]
\label{g+}
\end{equation}
where the $\rg^+_{ij}(x,r)$ are formal power series in the parameter $r$.

It has been proven \cite{FG} that a Poincar\'e metric for a given pair $(M,[\rg])$ exists and is unique, up to diffeomorphisms fixing $M$, for $n$ {odd} provided the $\rg^+_{ij}(x,r)$ are even functions of $r$.
If $n$ is even, $\Ric({\rg^+})+k{\rg^+}=\cO(r^{n-2})$ uniquely determines $\rg^+$ (modulo $\cO(r^n)$), again up to diffeomorphisms fixing $M$, for which the $\rg^+_{ij}(x,r)$ are even functions of $r$ (modulo $\cO(r^n)$).

\goodbreak

Let us illustrate this construction by the well-known example of the Einstein space $\Ein_{n-1,1}=\partial(\AdS_{n+1})$, the archetype of Poincar\'e metric being provided by the {Anti de-Sitter} ($\AdS$) metric whose conformal infinity is the Einstein conformal structure (compactified Minkowski space). Here, $M=\Ein_{n-1,1}$, and $M^+=\AdS_{n+1}$.

Start with $\bbR^{n+2}$, where $n=d+2$, with the following metric
\begin{equation}
G=\sum_{i=1}^d{dx^i\otimes{}dx^i}+2dx^{d+1}\odot{}dx^{d+2}+2dx^{d+3}\odot{}dx^{d+4}
\label{G}
\end{equation}
of signature~$(n,2)$, and consider the unit hyper\-boloid
\begin{equation}
{\AdS_{n+1}}=\{X\in\bbR^{n,2}\,\vert\,\overline{X}X=-1\}
\label{AdSn+1}
\end{equation}
where $\overline{X}\equiv{}G(X)$ as a shorthand notation. 
The induced metric 
$
{\rg^+}=G\strut\vert_{T\AdS_{n+1}}
$
is Lorentzian (of signature $(n,1)$), and of constant sectional curvature.
View now ${\AdS_{n+1}}\subset\bbP^{n+1}(\bbR)$ as the projectivized open ball $\cB=
\{X\in\bbR^{n,2}\,\vert\,\overline{X}X<0\}$. Its conformal boundary is the Einstein space 
\begin{equation}
{\Ein_{n-1,1}}=\bbP\cQ
\label{Einsn-1,1}
\end{equation}
which is the projectivization of the null cone $\cQ=\{Q\in\bbR^{n, 2}\setminus\{0\}\,\vert\,\overline{Q}Q=0\}\cong\bbR^*_+\times(S^{n-1}\times{}S^1)$, and is endowed with the conformal class $[\rg]$ of Lorentzian metrics in\-herited from $\strut{}G\vert_{T\cQ}$; see, e.g., \cite{Fra}.
Conformal infinity of $(\AdS_{n+1},{\rg^+})$ is therefore
$(\Ein_{n-1,1},[\rg])$, and both
\begin{eqnarray}
\label{AdS}
\AdS_{n+1}&\cong&\bbR^n\times{}S^1,\\
\label{Ein}
\Ein_{n-1,1}&\cong&(S^{n-1}\times{}S^1)/\bbZ_2
\end{eqnarray}
are homogeneous spaces of $\rO(n,2)$. 

We refer to \cite{Gib} for a comprehensive review of the geometry of $\AdS$ spacetimes.

\subsection{Schr\"odinger manifolds \& conformal Bargmann structures as null infinity}\label{SchManif}

In the wake of the previously reviewed work \cite{FG}, we will  
introduce the new notion of a (Poincar\'e-)Schr\"odinger manifold whose ``confor\-mal infinity'' is a given conformal Bargmann structure $(M,[\rg],\xi)$
in the sense of Definition \ref{ConfBargStrDef}. 

\begin{defi}\label{SchManifDef}
A {(Poincar\'e-)Schr\"odinger manifold} for a conformal {Barg\-mann structure} $(M,[\rg],\xi)$ is a triple $(\hM,\hrg,\hxi)$ with $M=\partial\hM$, where $\hrg$ is a Lorentz metric on $\Int(\hM)$, and~$\hxi$ a nowhere vanish\-ing lightlike Killing vector field for $\hrg$ such that
\begin{enumerate}
\item\label{one}
$\hxi\,\big\vert_{T^*M}=\xi$
\item\label{two}
$\hrg^{\,-1}\big\vert_{T^*M}=\mu\,\xi\otimes\xi$ (normalization condition: $\mu=1$)
\item\label{three}
$\rg^+=\hrg+\mu\,\htheta\otimes\htheta$ (where
$\htheta=\hrg(\hxi)$) is a {Poincar\'e metric} for $(M,[\rg])$.
\end{enumerate}
\end{defi}

Let us explain and justify the different items of Definition \ref{SchManifDef}.

The generator $\xi$ of the structure group of the principal $H$-bundle $\pi:M\to{}M/H$ extends smoothly to $\hM$. In fact, the null vector field $\xi$ enters the definition of the character of $H$ associated with the mass in the Schr\"odinger equation~(\ref{SchEq}); as such, it ought to give rise to a unique nowhere vanishing, null, Killing vector field $\hxi$ for $\hrg$. 
This justifies our first axiom. 

In Axiom 2, the real constant $\mu$ is, in fact, quite arbitrary; it allows, via the null vector field $\xi$, for extra terms in the metric $\hrg$ with higher order singularities 
at the conformal boundary, $M$.\footnote{See, e.g., Equations (\ref{hthetaBis}) and (\ref{hgGeneral}) below showing that the metric $\hrg$ exhibits a singular behavior $\sim{}r^{-4}$ at conformal infinity, namely $r^4\,\hrg\big\vert_{TM}\in[\theta\otimes\theta]$.
} The normalization condition is dictated by the form (\ref{BSmetric}) of the metric dealt with in the literature about non-relativisic AdS/CFT correspondence. 


The third axiom, in Definition \ref{SchManifDef}, resorts explicitly to the conformal class of Bargmann metrics; it is thus devised to make use of the FG approach to Poincar\'e metrics which is at our disposal (see Definition \ref{FG}). However, we will not address here the problem of the existence and uniqueness of (Poincar\'e-)Schr\"o\-dinger structures of Definition \ref{SchManifDef}. Instead, we will provide explicit examples. 
We duly recover the FG axioms if $\hxi$ is ignored, or if $\mu=0$.

\begin{pro}\label{gmuPro}
Let us set $\htheta=\hrg(\hxi)$, then
the family of symmetric tensor fields 
$\trg=\hrg+\mu\,\htheta\otimes\htheta$ para\-metrized by $\mu\in\bbR$ defines on $\hM$ a family of Lorentzian metrics for which $\txi=\hxi$ is a null, nowhere vanishing, Killing vector field.
\end{pro}
\begin{proof}
If $\{\lambda_1,\ldots,\lambda_{d+1},+\lambda,-\lambda\}$ denotes the spectrum of the Gram matrix of $\hrg$ with respect to some basis, then the spectrum of the corresponding Gram matrix of $\trg$ is given by $\{\lambda_1,\ldots,\lambda_{d+1},\half\mu+\sqrt{\lambda^2+\left(\half\mu\right)^2},\half\mu-\sqrt{\lambda^2+\left(\half\mu\right)^2}\}$, with the same Lorentz signature. Since $\hxi$ is null for $\hrg$, i.e., $\htheta(\hxi)=0$, then $\txi=\hxi$ is clearly $\trg$-null. Moreover, the fact that $\hxi$ is a Killing vector field for~$\hrg$ entails that the same is true for~$\trg$.
\end{proof}

Had we put $\txi=\alpha\,\hxi$ for some $\alpha\in{}C^\infty(\hM,\bbR^*)$ in Proposition \ref{gmuPro}, we would have found that necessarily $d\alpha=0$, i.e., $\alpha\in\bbR^*$.

We will resort to Proposition \ref{gmuPro} in Section \ref{HomSchManifSection} where we  supply paragons of Schr\"odinger manifolds.


\section{Global structure of the Schr\"odinger group}\label{SchLieGroupSection}

\subsection{Flat conformal Bargmann structure and Schr\"odinger Lie algebra}\label{SchLieAlgSection}

The conformal automorphisms of a Bargmann structure $(M,\rg,\xi)$ --- which we will later on identify, for the flat structure (\ref{FlatBargmannStructure}), to the so-called ``Schr\"odinger group'' --- have been introduced in Proposition \ref{SymSchPro}. As read off Equations (\ref{SchroedingerAutomorphism}), they consist in (local) diffeo\-morphisms, $\Phi$, of~$M$ such that $\Phi^*\rg=\Omega_\Phi^2\cdot\rg$ and $\Phi_*\xi=\xi$ for some smooth, positive, function~$\Omega_\Phi$.
The latter turns out to be necessarily (the pull-back of) a function of the time axis, $T$. See Equation (\ref{EqRel}).


Accordingly, at the Lie algebraic level, the infinitesimal conformal auto\-mor\-phisms of such a structure span the so-called \textit{Schr\"odinger Lie algebra}, which is therefore the Lie algebra of those vector fields $Z$ of $M$ such that
\begin{equation}
\sfL_Z\,\rg=\varphi_Z\cdot\rg
\qquad
\&
\qquad
\sfL_Z\,\xi=0
\label{SchroedingerLieAlgebra}
\end{equation}
for some smooth function $\varphi_Z$, again necessarily defined on $T$.

The Schr\"odinger Lie algebra, $\sch(d+1,1)$, of the flat Bargmann structure (\ref{FlatBargmannStructure}) is therefore isomorphic to the Lie algebra of all smooth vector fields $Z=[x\mapsto\delta{x}]$ of $\bbR^{d+2}$ satisfying~(\ref{SchroedingerLieAlgebra}), i.e., of vector fields of the form
\begin{equation}
\delta{x}=\Lambda{}x
+
\Gamma
-\half\alpha\,\rg(x,x)\xi+\alpha\,\rg(\xi,x)x+\chi{}x
\label{schd+1,1}
\end{equation}
where $\Lambda\in\rso(d+1,1)$, $\Gamma\in\bbR^{d+2}$, and $\alpha,\chi\in\bbR$ are such that $\Lambda\xi+\chi\xi=0$. We find that $\dim(\sch(d+1,1))=\half(d^2+3d+8)$, so that $\dim(\sch(4,1))=13$ in the standard case $d=3$. Homogeneous Galilei transformations are generated by $\Lambda$, Bargmann translations by $\Gamma$, while $\alpha$ and $\chi$ generate inversions and dilations, respectively. 
The centre, $\rh\cong\bbR$, of $\sch(d+1,1)$ is generated by ``vertical'' translations~$\Gamma$, i.e., such that $\xi\wedge\Gamma=0$. The quotient $\sch(d+1,1))/\rh$ acts therefore on Galilei spacetime $E\cong\bbR^{d+1}$ as the Lie algebra of flat NC infinitesimal automorphisms; it is sometimes called the \textit{center-free} Schr\"odinger Lie algebra, and is isomorphic to $(\rso(d)\times\Sl(2,\bbR))\ltimes(\bbR^d\times\bbR^d)$.

\subsection{Schr\"odinger group as a subgroup of conformal group}
\label{SchLieGroupSection2}

Taking advantage of the content of the preceding section, let us focus attention on the global structure of the Schr\"odinger group, $\Sch(d+1,1)$, of the flat (conformal) Bargmann structure (\ref{FlatBargmannStructure}). The latter will be naturally chosen so as to integrate $\sch(d+1,1)$ \textit{inside} the ``conformal group'' of $\bbR^{d+1,1}$. 

Therefore, in view of (\ref{SchroedingerAutomorphism}), we will characterize the \textit{Schr\"odinger group} as a subgroup of the group, $\rO(d+2,2)$, of all linear isometries of $\bbR^{d+2,2}=\bbR^{d+1,1}\oplus\bbR^{1,1}$ endowed with the metric (\ref{G}) that we split according to
\begin{equation}
G=\left[\sum_{i=1}^d dx^i\otimes{}dx^i+2dx^{d+1}\odot{}dx^{d+2}\right]+2dx^{d+3}\odot{}dx^{d+4}
\label{G1}
\end{equation}
in order to render explicit the Bargmann metric $\rg$ as given by (\ref{FlatBargmannStructure}); this metric reads in matrix guise,\footnote{The matrix (\ref{GramG}) is the Gram matrix of some chosen basis that will not be further specified, unless otherwise stated.} 
\begin{equation}
G=\pmatrix{
\rg&0&0\cr
0&0&1\cr
0&1&0
}.
\label{GramG}
\end{equation}

We need, at this stage, a new geometric object, namely a preferred element, $Z_0$, of the Lie algebra, $\ro(d+2,2)$, of $\rO(d+2,2)$.

\begin{defi}
We will call ``special null vector" any $Z_0\in\ro(d+2,2)$ such that: (i)~$(Z_0)^2=0$, and (ii) $Z_0\neq0$.
\end{defi}

The following Lemma is classical; see, e.g.,~\cite{JMS,GS}.

\begin{lem}\label{lemmaZ0}
A special null vector is of the general form $Z_0=P_0\wedge{}Q_0$ for
some $P_0,Q_0\in\bbR^{d+2,2}\!\setminus\!\{0\}$ such that
$G(P_0,P_0)=G(Q_0,Q_0)=G(P_0,Q_0)=0$.\footnote{We will often use the
  identification $\ro(d+2,2)\cong\bigwedge^2\bbR^{d+2,2}$.} The set of
these vectors form a single adjoint orbit of
$\rO(d+2,2)$.\footnote{This nilpotent orbit has two connected components; in the case $d=2$ each one is symplectomorphic to the manifold of regularized Keplerian
  motions \cite{JMS,GS,DET}.} 
\end{lem}

Our choice of origin of the orbit of special null vectors is performed by selecting $P_0=e_{d+2}$, and $Q_0=e_{d+3}$ where $e_i=\partial/\partial{x^{i}}$ for all $i=1,\ldots,d+4$. It thus reads
\begin{equation}
\medbox{
Z_0=\pmatrix{
0&0&\xi\cr
-\xi^*&0&0\cr
0&0&0
}
\in\ro(d+2,2)}
\label{Z0}
\end{equation}
where $\xi\in\bbR^{d+2}\setminus\{0\}$ is as in
(\ref{FlatBargmannStructure}), the superscript ``$*$'' standing for
the $\rg$-adjoint; thus, $\xi^*=\rg(\xi)$ is the
co\-vector~$\xi^*=\theta(=dt)$, interpreted as the \textit{Galilei
  clock} (see Section \ref{ConfBargSection}). This~$Z_0$ will henceforth be identified with the null generator,
$\xi$, of ``vertical translations''  on Bargmann space~$\bbR^{d+1,1}$ \cite{Duv1,HH}.  
\begin{pro}\label{DefSchStabPro}
 The Lie algebra $\sch(d+1,1)$ is isomorphic to the Lie algebra of the group
\begin{equation}
\medbox{
\Sch(d+1,1)=\{A\in\rO(d+2,2)\,\vert\,AZ_0=Z_0A\}
}
\label{DefSch}
\end{equation}
which we call the ``Schr\"odinger group''.
\end{pro}

\begin{proof}
Straightforward computation shows that the stabilizer of $Z_0$ in $\rO(d+2,2)$ consists of matrices of the form
\begin{equation}
A=\pmatrix{
L&a\xi&C\cr
B^*&b&d\cr
-a\xi^*&0&e}
\label{A}
\end{equation}
where $L\in{}\End(\bbR^{d+2})$, $B,C\in\bbR^{d+2}$, and $a,b,d,e\in\bbR$ satisfy
\begin{eqnarray}
\label{Lxi-exi=0}
0&=&L\xi-e\xi\\
0&=&L^*\xi-b\xi\\
\label{L*L=dots}
\bone&=&L^*L-a(\xi B^*+B\xi^*)\\
\label{L*C-adxi+eB=0}
0&=&L^*C-ad\xi+eB\\
\label{1=axi*C+be}
1&=&a\xi^*{}C+be\\
0&=&\xi^*(B+C)\\
\label{C2+2de=0}
0&=&C^*C+2de
\end{eqnarray}
where, again, $L^*$ stands for the $\rg$-adjoint of the linear operator $L$.

In view of (\ref{FlatBargmannStructure}) and (\ref{GramG}), let us put $\xi=e_{d+2}$, where $(e_1,\ldots,e_{d+2})$ is the ``canonical'' basis of $\bbR^{d+2}$; let us complete it in $\bbR^{d+2}\oplus\bbR^2$ with the canonical basis $(e_{d+3},e_{d+4})$ of $\bbR^2$. Define then (with a slight abuse of notation) $A_i=Ae_i$ for all $i=1,\ldots,d+4$, where $A$ is as in~(\ref{A}). 
Upon specifying
\begin{equation}
X=A_{d+4}=\pmatrix{C\cr d\cr e},
\qquad
Y=A_{d+3}=\pmatrix{a\xi\cr b\cr0},
\label{XY}
\end{equation}
we trivially check that 
\begin{equation}
\oX X=\oY Y=\oX Y-1=0
\qquad
\&
\qquad
Z_0Y=0,
\label{RelXY}
\end{equation}
where $\oX=G(X)$ is, as before, the $G$-adjoint of $X\in\bbR^{d+2,2}$.

\goodbreak

The group law of $\Sch(d+1,1)$, plainly given by matrix multiplication using~(\ref{A}), translates as the group action $\Sch(d+1,1)\ni{}A:(X,Y)\mapsto(X',Y')$ given by
\begin{equation}
(X',Y')=(AX,AY)
\label{AXAY}
\end{equation}
on the $(d+4)$-dimensional manifold 
defined by the constraints (\ref{RelXY}).

We then find, using (\ref{A}), that vectors in the Lie algebra of $\Sch(d+1,1)$ are of the form
\begin{equation}
Z=\pmatrix{
\Lambda&\alpha\xi&\Gamma\cr
-\Gamma^*&\chi&0\cr
-\alpha\xi^*&0&-\chi
}
\label{cZ}
\end{equation}
where $\Lambda\in\rso(d+1,1)$, $\Gamma\in\bbR^{d+2}$, and $\alpha,\chi\in\bbR$ are such that $\Lambda\xi+\chi\xi=0$ (see~(\ref{Lxi-exi=0})).

Let us now prove that the Lie algebra of $\Sch(d+1,1)$ is indeed isomorphic to $\sch(d+1,1)$, whose action on flat Bargmann space is given by (\ref{schd+1,1}).

Assuming $e(=\oX{}Q_0)\neq0$, in view of (\ref{1=axi*C+be}), (\ref{C2+2de=0}), and (\ref{XY}) we can write 
\begin{equation}
X
=
\frac{1}{r}
\pmatrix{
x\cr
-\half{}x^*x\cr
1
},
\qquad
Y=
r
\pmatrix{
q\xi\cr
1-q\xi^*{}x\cr
0
}
\label{XYbis}
\end{equation}
where $x=C/e\in\bbR^{d+2}$, $q=ae\in\bbR$, and $r=1/e\in\bbR^*$. 
We deduce from (\ref{AXAY}) that the Schr\"odinger group acts projectively on \textit{Bargmann space} $\bbR^{d+1,1}$ according to $A:x\mapsto{}x'$, viz.,
\begin{equation}
x'=\frac{Lx-\half{}a(x^*x)\xi+C}{e-a\xi^*{}x}
\label{x'}
\end{equation}
where $A\in\Sch(d+1,1)$ is as in (\ref{A}). 
We, likewise, get the transformation law
\begin{equation}
r'=\frac{r}{e-a\xi^*{}x}
\label{r'}
\end{equation}
with the same notation as before. 

\goodbreak

As for the infinitesimal action of the Schr\"odinger group on $\bbR^{d+1,1}$, it can be computed, using (\ref{x'}), by $\delta{x}=\delta{x'}\big\vert_{A=\bone,\delta{A}=Z}$, where $Z$ is as in (\ref{cZ}); we then find 
$
\delta{x}=\Lambda{}x
+
\Gamma
-\half\alpha(x^*x)\xi+\alpha(\xi^*{}x)x+\chi{}x
$,
which exactly matches Equation (\ref{schd+1,1}). Note that we get from (\ref{r'}) $\delta{r}=(\alpha\xi^*{}x+\chi)r$. 

The proof that the Lie algebra of $\Sch(d+1,1)$ is isomorphic to $\sch(d+1,1)$ is complete.
\end{proof}

\begin{pro}
The Schr\"odinger group (\ref{DefSch}) has two connected components, 
\begin{equation}
\pi_0(\Sch(d+1,1))=\bbZ_2.
\label{pi0}
\end{equation}
\end{pro}
\begin{proof}
Let us express the matrix $Z_0'$ of the central element $Z_0$ given by (\ref{Z0}) in a new basis of $\bbR^{d+2,2}$ whose Gram matrix is
\begin{equation}
G'=\pmatrix{
\bone_{
\bbR^d}&0&0\cr
0&D&0\cr
0&0&D
}
\label{G'}
\end{equation}
where $D=\diag(1,-1)$. 
The sought expression is therefore
\begin{equation}
Z_0'=\pmatrix{
0&0&0\cr
0&0&U\cr
0&V&0
}
\label{Z0'}
\end{equation}
where
\begin{equation}
U=\half
\left(\begin{array}{rr}
1&-1\cr
-1&1
\end{array}
\right)
\qquad
\&
\qquad
V=-\half\pmatrix{
1&1\cr
1&1
}.
\label{UV}
\end{equation}
The group $\rO(d+2,2))$ has four connected components, and the generators $\{I,P,T,PT\}$ of $\pi_0(\rO(d+2,2))
=\pi_0(\rO(d+2))\times\pi_0(\rO(2))
\cong\bbZ_2\times\bbZ_2$ can be defined --- up to conjugation --- by
\begin{equation}
I=\pmatrix{
\bone_{
\bbR^d}&0&0\cr
0&\bone_{
\bbR^2}&0\cr
0&0&\bone_{
\bbR^2}
}\!,
\hspace{2mm}
P=\pmatrix{
S&0&0\cr
0&\bone_{
\bbR^2}&0\cr
0&0&\bone_{
\bbR^2}
}\!,
\hspace{2mm}
T=\pmatrix{
\bone_{
\bbR^d}&0&0\cr
0&\bone_{
\bbR^2}&0\cr
0&0&D
}
\label{IPTPT}
\end{equation}
where $S\in\rO(d)$ is such that $S^2=\bone_{\bbR^d}$ and $\det(S)=-1$.

\goodbreak

It is a trivial matter to check that the only non-zero commutators are $[T,Z_0']$ and $[PT,Z_0']$, proving, via the definition (\ref{DefSch}) of the Schr\"odinger group, that, indeed, $\pi_0(\Sch(d+1,1))$ is generated by $I$, and $P$.
\end{proof}

\subsection{A nilpotent coadjoint orbit of the conformal group}\label{CoadSection}

We highlight that the (non-relativistic) Schr\"odinger group is, interestingly, as\-sociated with a special homogeneous \textit{symplectic} manifold of the (relativistic) conformal group.

As we have seen in Proposition \ref{DefSchStabPro}, the Schr\"odinger group, $\Sch(d+1,1)$, is the stabilizer of~$Z_0\in\so(d+2,2)$, given by (\ref{Z0}), for the adjoint action of $\rO(d+2,2)$. The (co)adjoint orbit 
\begin{equation}
\medbox{
\cO_{Z_0}=\rO(d+2,2)/\Sch(d+1,1)
}
\label{cO}
\end{equation}
is therefore a $2(d+1)$-dimensional \textit{symplectic manifold} we now describe as follows.

Consider the left-invariant Maurer-Cartan $1$-form $\Theta=A^{-1}dA$, and the $1$-form $\varpi=-\half\Tr(Z_0\Theta)$ of $\rO(d+2,2)$. A classical result tells us that $d\varpi$ descends to $\cO_{Z_0}$ as the canonical Kirillov-Kostant-Souriau symplectic $2$-form, $\omega$, of~$\cO_{Z_0}$, viz., $d\varpi=(\rO(d+2,2)\to\cO_{Z_0})^*\omega$. Indeed, let us put again $\xi=e_{d+2}$, and $A_i=Ae_i$ for $i=1,\ldots{}d+4$ whenever $A\in\rO(d+2,2)$; with the help of (\ref{Z0}) we get $\varpi=\oP{}dQ$, where $P=A_{d+2}$, and $Q=A_{d+3}$ are nonzero, and such that $\oP{}P=\oP{}Q=\oQ{}Q=0$.
The $2$-form
\begin{equation}
d\varpi=d\oP\wedge{}dQ\label{dvarpi}
\end{equation}
clearly descends to the slit null tangent bundle, $NT\cQ\!\setminus\!\cQ$, of the null quadric 
\begin{equation}
\cQ=\{Q\in\bbR^{d+2,2}\setminus\!\{0\}\,\vert\,\oQ{}Q=0\}.
\label{Q}
\end{equation}
It defines the sought symplectic structure, $\omega$, on 
\begin{equation}
\cO_{Z_0}=(NT\cQ\!\setminus\!\cQ)/\SL(2,\bbR)
\label{cObis}
\end{equation}
interpreted as the \textit{manifold of null geodesics of conformally compactified Minkowski spacetime} $\bbP\cQ=\Ein_{d+1,1}$; see (\ref{Einsn-1,1}) and (\ref{Ein}). (The leaves of the distribution $\ker(d\varpi)$ of $NT\cQ\!\setminus\!\cQ$ project to $\bbP\cQ$ as the null geodesics of its conformally flat structure.)
Note that $\cO_{Z_0}=\cO^+_{Z_0}\cup\cO^-_{Z_0}$ with $\cO^\pm_{Z_0}\cong{}TS^{d+1}\!\setminus{}S^{d+1}$, topologically \cite{JMS,GS,DET}.

\section{Homogeneous Schr\"odinger manifolds}
\label{HomSchManifSection}

We are now led to the following query: 
what $\Sch(d+1,1)$-homogeneous space would host a \textit{genuine}, well-behaved, Lorentz metric whose isometries constitute the whole Schr\"odinger group (\ref{DefSch})?

Let us first consider the distinguished element $Z_0\in\so(d+2,2)$ represented as in (\ref{Z0}) and the associated vector field $\delta_{Z_0}:Q\mapsto{}Z_0Q$ on the quadric 
\begin{equation}
\AdS_{d+3}(\sqrt{-2\lambda})=\{Q\in\bbR^{d+2,2}\,\vert\,\oQ Q=2\lambda\}
\label{AdSd+3}
\end{equation}
with a given $\lambda<0$ (see (\ref{AdSn+1})).
\begin{lem}\label{Lem-deltaZ0neq0}
The vector field $\delta_{Z_0}$ of $\AdS_{d+3}(\sqrt{-2\lambda})$ nowhere vanishes.
\end{lem}
\begin{proof}
In view of (\ref{Z0}), we find
\begin{equation}
\delta_{Z_0}:
\pmatrix{
x\cr
\alpha\cr
\beta
}
\mapsto
\pmatrix{
\beta\xi\cr
-\xi^*x\cr
0
}
\label{deltaZ0}
\end{equation}
where $x\in\bbR^{d+1,1}$, and $\alpha,\beta\in\bbR$ are such that $\oQ{}Q=x^*x+2\alpha\beta=2\lambda$, and where the metric (\ref{G}) has been used. Suppose, for the moment, that $\delta_{Z_0}Q=0$ for some $Q\in\AdS_{d+3}(\sqrt{-2\lambda})$, i.e., that $\beta=0$, and $\xi^*x=0$. We readily get $x^*x=2\lambda<0$. We hence find that $x\in\bbR^{d+1,1}$ is at the same time $\rg$-orthogonal to the null vector $\xi\neq0$, and timelike: contradiction! Thus, $\delta_{Z_0}Q\neq0$ for all $Q\in\AdS_{d+3}(\sqrt{-2\lambda})$.
\end{proof}

\subsection{A special family of Schr\"odinger-homogeneous spaces}

Let us resort to the definition (\ref{XY}) of the vectors $X,Y\in\cQ$ (the last two column vectors of the Schr\"odinger matrix (\ref{A})), and posit 
\begin{equation}
\medbox{
Q{}=X+\lambda Y
}
\label{Qlambda}
\end{equation}
where $\lambda\in\bbR^*$ is fixed. 

We contend, and will prove right below, that the set 
\begin{equation}
\medbox{
\hM_\lambda=\{X+\lambda{}Y\in\bbR^{d+2,2}\,\vert\,\oX X=\oY Y=\oX Y-1=0, Z_0Y=0
\}
}
\label{hatM}
\end{equation}
of these $Q{}$, with $\lambda<0$, is actually a homogeneous manifold of the Schr\"odinger group, and an open submanifold
$\hM_\lambda\subset\AdS_{d+3}(\sqrt{-2\lambda})$. 

\goodbreak

\begin{pro}\label{ProHomSpace}
For every $\lambda<0$, the manifold (\ref{hatM}) is a connected, $(d+3)$-dimensional, homogeneous space of the Schr\"odinger group, viz., 
\begin{equation}
\medbox{
\hM_\lambda\cong\Sch(d+1,1)/(\rE(d)\times\bbR)
}
\label{Mlambda0Bis}
\end{equation}
where $\rE(d)=\rO(d)\ltimes\bbR^d$ is the Euclidean group of $\bbR^d$. These manifolds have topology
\begin{equation}
\medbox{
\hM_\lambda\cong(\bbR^{d+2}\setminus\{0\})\times{}S^1.
}
\label{Mlambda0Ter}
\end{equation}
\end{pro}

\begin{proof}
From the very definition (\ref{hatM}), each manifold $\hM_\lambda$ is the image of the surjection $\pi_\lambda:\Sch(d+1,1)\to\hM_\lambda$ given by $\pi_\lambda(A)=\lambda{}A_{d+3}+A_{d+4}$. The left-action of the Schr\"odinger group clearly passes to the quotient according to (\ref{AXAY}), and $\hM_\lambda$ is therefore diffeomorphic to a homogeneous space $\Sch(d+1,1)/K$. Let us prove that $K\cong\rE(d)\times\bbR$.

The coordinate system chosen in (\ref{XYbis}) provides us with the local expression
\begin{equation}
Q{}
=
\frac{1}{r}
\pmatrix{
\hx\cr
-\half{}\hx^*\,\hx+\lambda r^2\cr
1
}
\label{localQ}
\end{equation}
where 
\begin{equation}
\hx=x+\lambda r^2 q\xi\in\bbR^{d+2}
\label{hx}
\end{equation}
and $r\neq0$. 
Consider now the ``origin'', $Q{}_0$,
defined by $\hx=0$ and $r=1$ in (\ref{localQ}). Then look for the subgroup, $K$, 
of all $A\in\Sch(d+1,1)$ such that $AQ{}_0=Q{}_0$. In view of (\ref{A}) we readily find $C=-\lambda a \xi$, $d=\lambda(1-b)$, and $e=1$. Moreover, the constraints (\ref{Lxi-exi=0})--(\ref{C2+2de=0}) yield $L\xi=L^*\xi=\xi$, $B=-C=\lambda a \xi$, $b=1$, $d=0$, and $L^*L=\bone+2\lambda{}a^2\xi\xi^*$. The equations $L\xi=L^*\xi=\xi$ help us write
\begin{equation}
L=\pmatrix{
R&u&0\cr
0&1&0\cr
v&w&1}
\label{L}
\end{equation}
where $R\in{}\End(\bbR^d)$, $u\in\bbR^d$, $v\in(\bbR^d)^*$, and $w\in\bbR$. At last, the extra constraint $L^*L=\bone+2\lambda{}a^2\xi\xi^*$ entails that $R^tR=\bone$, $v=-u^tR$ (where the superscript ``$t$'' stands for transposition), and $w=-\half{}u^tu+\lambda{}a^2$. The isotropy group $K\subset\Sch(d+1,1)$ of $Q{}_0$ is therefore parametrized by the triples $(R,u,a)\in{}\rO(d)\times\bbR^d\times\bbR$, and easily found to be isomorphic to the direct product $K\cong\rE(d)\times\bbR$. Since $\dim(K)=\half{}d(d+1)+1$, we indeed get $\dim(\hM_\lambda)=\half(d^2+3d+8)-\half(d^2+d+2)=d+3$.

We now work out the topology of our $\Sch(d+1,1)$-homogeneous space $\hM_\lambda$ given by (\ref{hatM}). We will suitably use a frame of $\bbR^{d+2,2}$ with Gram matrix (\ref{G'}) where the distinguished element $Z_0\in\ro(d+2,2)$ is represented by the matrix $Z'_0$ in~(\ref{Z0'}).

Let us write the components of the $Q{}=X+\lambda Y$, defined by (\ref{Qlambda}), in this frame. Solving the equations in (\ref{hatM}) for $Y$, namely $\oY Y=0$, and $Z_0Y=0$, we get
$$
Y=\pmatrix{
0\cr
a'\cr
-a'\cr
b'\cr
b'
}\in\bbR^{d+2,2}\setminus\!\{0\}
$$
with $a',b'\in\bbR$ (and $a'^2+b'^2>0$). 

As for the remaining equations satisfied by
$$
X=\pmatrix{
x\cr
a\cr
u\cr
b\cr
v
}\in\bbR^{d+2,2}\setminus\!\{0\}
$$
with $x\in\bbR^d$, and $a,b,u,v\in\bbR$, we obtain
\begin{eqnarray}
\label{X2=0}
\oX X=0 &\Longleftrightarrow& x^tx+a^2+b^2=u^2+v^2>0\\
\label{XY=1}
\oX Y=1 &\Longleftrightarrow& a'(u+a)+b'(b-v)=1.
\end{eqnarray}
Noticing that the dilations $(X,Y)\mapsto(\alpha X,\alpha^{-1}Y)$ with $\alpha\in\bbR^*$ do preserve $\hM_\lambda$, we claim that the latter dilation invariance and the  conditions (\ref{X2=0}) and~(\ref{XY=1}) leave us with $d+6-3=d+3$ free parameters, e.g., $x,a,b,u,v,a'$. Then, Equation~(\ref{XY=1}) yields $a'$ as a function of $a,b,u,v$. The only remaining constraint on $X$ is therefore given by Equation (\ref{X2=0}). This entails that $X\in\cQ$, hence~$\hM_\lambda$ has the same topology as $\cQ\cong(\bbR^{d+2}\setminus\!\{0\})\times{}S^1$, and is thus connected. 
\end{proof}


\subsection{Distinguished Schr\"odinger-invariant structures}\label{LorentzStrSubSec}

With these preparations, we are ready to introduce Schr\"odinger-invariant tensors on~$\hM_\lambda$.

Denote by $\hrg_\lambda=(\hM_\lambda\hookrightarrow\bbR^{d+4})^*G$ the induced symmetric tensor on $\hM_\lambda$, viz.,
\begin{equation}
\medbox{
\hrg_\lambda(\delta Q,\delta'Q)=\delta\oQ\,\delta'Q
}
\label{hg}
\end{equation}
for all $\delta Q,\delta'Q\in{}T_Q\hM_\lambda$. This tensor, $\hrg_\lambda$, is clearly $\Sch(d+1,1)$-invariant. In view of~(\ref{DefSch}), the same remains true for the one-form $\htheta$ of $\hM_\lambda$ defined by
\begin{equation}
\medbox{
\htheta(\delta{Q})=-\oQ\,Z_0\,\delta{Q}
}
\label{htheta}
\end{equation}
for all $\delta{Q}\in{}T_Q\hM_\lambda$.

We easily find that $d\htheta(\delta{Q},\delta'{Q})=-2\delta\oQ Z_0\delta'{Q}$. By means of the fact that $Z_0$ has rank $2$ (as clear from Lemma \ref{lemmaZ0} stating that $Z_0=P_0\wedge{}Q_0$ where $P_0$ and $Q_0$ span a totally null plane in $\bbR^{d+2,2}$), and by some straightforward computation, we get
\begin{equation}
\htheta\wedge{}d\htheta=0.
\label{thetadtheta=0}
\end{equation}


\begin{rmk}
{\rm
\textit{Local} expressions for (\ref{hg}) and (\ref{htheta}) are easily deduced from~(\ref{localQ}); we get 
$\hrg_\lambda(\delta Q,\delta'Q)=r^{-2}\left[\rg(\delta\hx,\delta'\hx) -2\lambda\,\delta r\delta'r\right]$
or, alternatively,
\begin{equation}
\hrg_\lambda=\frac{1}{r^2}\left[\sum_{i,j=1}^{d+2}\rg_{ij}\,d\hx^i\otimes{}d\hx^j-2\lambda\,dr\otimes{}dr\right]
\label{localhgBis}
\end{equation}
together with
\begin{equation}
\htheta=\frac{\theta}{r^2}
\label{hthetaBis}
\end{equation}
where $\theta=\sum_{i=1}^{d+2}\rg_{ij}\,\xi^i d\hx^j(=dt)$ is the Galilei clock of the flat Bargmann structure.
The metric (\ref{localhgBis}) is the well-known expression of the $\AdS_{d+3}(\sqrt{-2\lambda})$ metric in Poincar\'e coordinates; see, e.g., \cite{AGMOO}.
}
\end{rmk}

\begin{thm}\label{mainThm0}
For every $\lambda<0$, the manifold $\hM_\lambda$ admits a family of Lorentz metrics
\begin{equation}
\medbox{
\hrg_{\lambda,\mu}=\hrg_\lambda-\mu\,\htheta\otimes\htheta
}
\label{hgGeneral}
\end{equation}
given by (\ref{hg}) and (\ref{htheta}), parametrized by $\mu\in\bbR$. The Schr\"odinger group is the group of isometries of $(\hM_\lambda,\hrg_{\lambda,\mu})$.
\end{thm}
\begin{proof}
The signature of the metrics $\hrg_{\lambda}$ and $\hrg_{\lambda,\mu}$ is clearly Lorentzian since $\lambda<0$. 
Then the group of isometries of $(\hM_\lambda,\hrg_{\lambda})$ is, by construction, a subgroup of the group $\rO(d+2,2)$ of isometries of $\AdS_{d+3}(\sqrt{-2\lambda})$, which furthermore preserves the constraint $Z_0Y=0$ in (\ref{hatM}). It is thus the stabilizer of $Z_0$ in $\rO(d+2,2)$, i.e., the Schr\"odinger group $\Sch(d+1,1)$ in view of (\ref{DefSch}). The extra term, $-\mu\,\htheta\otimes\htheta$, in~(\ref{hgGeneral}) being $\Sch(d+1,1)$-invariant,  Proposition \ref{gmuPro} helps us complete the proof.
\end{proof}

\begin{rmk}
{\rm
The expression (\ref{hgGeneral}) is --- up to an overall multiplicative constant factor --- the most general twice-symmetric tensor constructed by means of the only data at our disposal, namely the ``ambient'' metric $G$ given by (\ref{G}), and the central element $Z_0\in\sch(d+1,1)$ defined in (\ref{Z0}).  
}
\end{rmk}

\begin{rmk}
{\rm
In view of Proposition \ref{ProHomSpace}, the manifold (\ref{hatM}) is $(d+3)$-dimensional, it is thus an open submanifold $\hM_\lambda\subset\AdS_{d+3}(\sqrt{-2\lambda})$.
}
\end{rmk}

There exists a privileged vector field on $\hM_\lambda$, namely
\begin{equation}
\medbox{
\hxi:Q\mapsto
\delta_{Z_0}Q=Z_0Q
}
\label{deltaZ0Q}
\end{equation}
where $Z_0\in\ro(d+2,2)$ is defined by (\ref{Z0}). 
\begin{pro}\label{ProhxiKilling}
The vector field $\hxi$ defined by (\ref{deltaZ0Q}) is a nowhere vanishing, lightlike, Killing vector field of $(\hM_\lambda,\hrg_{\lambda,\mu})$.
\end{pro}
\begin{proof}
The restriction $\hxi$ to $\hM_\lambda$ of the vector field $\delta_{Z_0}:Q\mapsto{}Z_0Q$ of $\bbR^{d+2,2}$ is tangent to~$\hM_\lambda$ at the point $Q$ since $\delta_{Z_0}(\oQ{}Q)=2\oQ\delta_{Z_0}Q=2\oQ{}Z_0Q=0$ as a consequence of the $G$-skewsymmetry of~$Z_0$. Let us furthermore show that $Z_0Q\neq0$ for all $Q\in\hM_\lambda$. Resorting to (\ref{hatM}), we find $Z_0Q=Z_0X$; using (\ref{Z0}) and (\ref{XY}), we get
\begin{equation}
Z_0X=\pmatrix{e\xi\cr{}-\xi^*C\cr0}
\end{equation}
and claim that the latter vector nowhere vanishes since $\xi\neq0$. Indeed, suppose that~$e=0$; then Equation~(\ref{1=axi*C+be}) would necessarily yield $\xi^*C\neq0$, implying $Z_0X\neq0$, whence $\delta_{Z_0}Q\neq0$ for all $Q\in\hM_\lambda$.

The vector field (\ref{deltaZ0Q}) is actually a Killing vector field of the metric (\ref{hgGeneral}) since it generates the $1$-parameter additive group $s\mapsto\exp(s{}Z_0)=\Id+sZ_0\in\Sch(d+1,1)$, i.e., a group of isometries of $(\hM_\lambda,\hrg_{\lambda,\mu})$ as a consequence of Theorem \ref{mainThm0}.

We finally check that $\rg_{\lambda,\mu}(\hxi,\hxi)=0$. By Equations (\ref{hg}) and~(\ref{htheta}), we get $\rg_{\lambda,\mu}(\delta_{Z_0}Q,\delta_{Z_0}Q)=\overline{Z_0Q}\,Z_0Q-\mu(\oQ Z_0^2Q)^2=0$ since $Z_0+\overline{Z_0}=Z_0^2=0$.
\end{proof}

\subsection{Conformal infinity and conformal Bargmann structures}
\label{confInfBargSection}

Resorting to Definition (\ref{hatM}), we will consider the limit $\lambda\to0$ as a route to conformal infinity of $(\hM_\lambda,\hrg_{\lambda,\mu},\hxi)$, our candidate to the status of Schr\"odinger manifold.

Observe that, in view of Lemma \ref{Lem-deltaZ0neq0}, there holds $Z_0X\neq0$ in (\ref{hatM}). So, the limiting manifold $\hM_0=\lim_{\lambda\to0}{\hM_\lambda}$ is an open submanifold of the null cone $\cQ$. The construction (\ref{Einsn-1,1}) of the Einstein space therefore prompts the following definition for conformal infinity of the previous structure, namely $M=\hM_0\,/\,\bbR^*$, i.e.,
\begin{equation}
\medbox{
M
=
\{X\in\bbR^{d+2,2}\,\vert\,\oX{}X=0,Z_0X\neq0\}\,/\,\bbR^*
}
\label{M}
\end{equation}
where $X\sim{}X'$ iff $X'=\alpha{}X$ for some $\alpha\in\bbR^*$.

\goodbreak

\begin{pro}\label{ProConfBargHomSpace}
The manifold (\ref{M}) is diffeomorphic to the following $(d+2)$-dim\-ensional homo\-geneous space of the Schr\"odinger group
\begin{equation}
\medbox{
M\cong\Sch(d+1,1)/(\rE(d)\times T\bbR^*)
}
\label{MBis}
\end{equation}
and has topology
\begin{equation}
\medbox{
M\cong(\bbR^{d+1}\times{}S^1)/\bbZ_2.
}
\label{topM}
\end{equation}
\end{pro}
\begin{proof}
If $X\in\hM_0$, the same is true for $AX$ for any $A\in\Sch(d+1,1)$ since $Z_0AX=AZ_0X\neq0$, see Definition (\ref{DefSch}). This enables us to choose, e.g.,
\begin{equation}
X
=
\pmatrix{
0\cr
0\cr
1
}\in\hM_0
\label{localX}
\end{equation}
in the frame whose Gram matrix is as in (\ref{GramG}).

Now, $M$ being the projectivization of $\hM_0$, let us determine the stabilizer, $S$, of the direction of $X$ in (\ref{localX}). Seek thus the form of those $A\in\Sch(d+1,1)$ such that $AX=\alpha{}X$, for some $\alpha\in\bbR^*$. Using (\ref{A}), we get $C=0$, $d=0$, and $e=\alpha$.
Equations (\ref{L*C-adxi+eB=0}) and (\ref{1=axi*C+be}) entail $B=0$, and $b=1/e$. From Equation (\ref{L*L=dots}) we get $L^*L=\bone$, hence
\begin{equation}
A=\pmatrix{
L&a\xi&0\cr
0&e^{-1}&0\cr
-a\xi^*&0&e}
\label{ABis}
\end{equation}
with $L\in\rO(d+1,1)$ satisfying the constraint (\ref{Lxi-exi=0}), $a\in\bbR$, and $e\in\bbR^*$. 

In order to implement the latter constraint $L\xi=e\xi$, and fully characterize $A\in{}S$, let us choose the constant $\rg$-null vector $\xi$ to be of the form
\begin{equation}
\xi
=
\pmatrix{
0\cr
0\cr
1
}\in\bbR^{d+1,1}
\label{localxi}
\end{equation}
as in the coordinate system used in (\ref{FlatBargmannStructure}). This entails that 
\begin{equation}
L=\pmatrix{
R&\displaystyle-e^{-1}Rv&0\cr
0&e^{-1}&0\cr
v^t&-\frac{1}{2}e^{-1}v^tv&e
}
\label{LBis}
\end{equation}
with $R\in\rO(d)$, and $v\in\bbR^d$. The matrix group law for this stabilizer readily yields
$S=(\rO(d)\ltimes\bbR^d)\times(\bbR^*\ltimes\bbR)$
proving (\ref{MBis}). We can therefore confirm that $\dim(M)=\half(d^2+3d+8)-\half(d^2+d+4)=d+2$.

We now work out the topology of $M$, our $\Sch(d+1,1)$-homogeneous space~(\ref{M}). To that end, use a frame with Gram matrix (\ref{G'}) where the distinguished element $Z_0\in\ro(d+2,2)$ is represented by the matrix $Z'_0$ given by (\ref{Z0'}).
Then look for all $X\in\cQ\setminus\widehat{M}_0$, i.e., for those $X$ lying in the null cone $\cQ$, and outside $\widehat{M}_0$. This amounts to finding all solutions
$$
X=\pmatrix{
x\cr
a\cr
u\cr
b\cr
v
}\in\bbR^{d+2,2}\setminus\!\{0\}
$$
with $x\in\bbR^d$, and $a,b,u,v\in\bbR$ of both equation $\oX{}X=0$, viz.,
$x^tx+a^2+b^2=u^2+v^2$, 
and $Z_0X=0$, i.e., $u=-a$, and $v=b$.
Since $X\neq0$, we get $x^tx=0$. This leaves us with $x=0$, and
$a^2+b^2>0$; hence
$
\cQ\setminus\widehat{M}_0
\cong
(\{\mathsf{pt}\}\times{}S^1)\times\bbR^*_+,
$
which reveals that, in this forbidden domain, the fibre above $(a,b)\neq0$ is a point,~$\{\mathsf{pt}\}$. Thanks to  (\ref{Ein}), and~(\ref{M}), we  obtain $M\cong\left((S^{d+1}\setminus\{\mathsf{pt}\})\times{}S^1\right)/\bbZ_2$, i.e.,
$M\cong(\bbR^{d+1}\times{}S^1)/\bbZ_2$.
\end{proof}

\begin{rmk}
{\rm
As a consequence of (\ref{topM}), the manifold (\ref{M}) has the topology of a M\"obius band as shown in \cite{Duv,DH2}. It will be interpreted as an extended spacetime, fibered above the time axis $T\cong\bbP^1(\bbR)$; see Section~\ref{ConfBargSection}.
}
\end{rmk}

Let us show that $M$ is actually endowed with a conformal Bargmann structure (see Section \ref{ConfBargStrSubSection}) inherited from its very definition (\ref{M}).

Consider then $\hrg_\lambda=\hrg_{\lambda,0}$ where $\hrg_{\lambda,\mu}$ is as in (\ref{hgGeneral}). The induced twice symmetric covariant tensor field $\rg_0=\hrg_0\big\vert_{T\hM_0}$ 
on $\hM_0\subset\cQ$ is degenerate, and $\ker(\rg_0)$ is spanned by $\cE$, the restriction to $\hM_0$ of the Euler vector field of the quadric $\cQ$.
We find that $\sfL_\cE\,\rg_0=2\rg_0$, which entails that $\rg_0$ defines but a conformal class $[\rg]$ of Lorentz metrics on~$M=\bbP\hM_0$ (just as in the $\Ein_{d+1,1}$ case dealt with in Section \ref{sectionFG}).

We have thus proved the following result.

\begin{pro}
The quadratic form $\rg_0$ on $\hM_0$ defines a conformal class $[\rg]$ of Lorentz metrics on $M$.
\end{pro}

Let us derive, at this stage, a remarkable global representative of $[\rg]$ constructed via a nowhere vanishing function $F_0$ of $\hM_0$, which is homogeneous of degree $2$, e.g., via the function
\begin{equation}
F_0(X)=(\oX{}P_0)^2+(\oX{}Q_0)^2
\label{F}
\end{equation}
associated with the distinguished element $Z_0=P_0\wedge{}Q_0\in\ro(d+2,2)$ of Lemma \ref{lemmaZ0}. Indeed $F_0(X)\neq0$ is equivalent to the defining condition $Z_0X\neq0$ of $\hM_0$.

We will also denote by $\pi:\hM_0\to{}M$ the projection where $\pi(X)=[X]$ is the ray through~$X\in\hM_0$.

\begin{lem}\label{gY0in[g]}
A representative $\rg_{F_0}\in[\rg]$ associated with the choice (\ref{F}) reads as
\begin{equation}
\rg_{F_0}(\delta[X],\delta'[X])=\frac{\delta\oX\,\delta'X}{F_0(X)}.
\label{gY0}
\end{equation}
\end{lem}
\begin{proof}
Clearly, the quadratic form $\rg_0/F_0$ is dilation-invariant, and hence passes to the quotient $M$ as a representative $\rg_{F_0}\in[\rg]$. Putting $[X]=X/\sqrt{F_0(X)}$ for $X\in\hM_0$, and using the fact that $\oX{}X=0$, 
we end up with Equation (\ref{gY0}).
\end{proof}

Moreover the action of the Schr\"odinger group on $M$, given by $A:[X]\mapsto[AX]$ for all $A\in\Sch(d+1,1)$ is well-defined; we further check, via Equation (\ref{gY0}), that it is indeed a conformal action since it preserves $[\rg]$.

Considering then the $1$-form $\htheta_0$ induced by $\htheta$ on $\hM_0$, we find that $\htheta_0(\cE)=0$, and $\sfL_\cE\,\htheta_0=2\htheta_0$. This implies, with the above choice, that the dilation-invariant $1$-form $\htheta_0/F_0$ descends to $M$ as the $1$-form $\theta_{F_0}$ given by
\begin{equation}
\theta_{F_0}(\delta[X])=-\frac{\oX{}Z_0\delta{X}}{F_0(X)}.
\label{thetaGlobal}
\end{equation}
Let us then prove that $\theta_{F_0}$ is closed. As a first step, we obtain
$d\theta_{F_0}(\delta[X],\delta'[X])=\left(-2F_0(X)\,\delta\oX{}Z_0\delta'{X}+\delta{F_0(X)}\oX{}Z_0\delta'{X}-\delta'{F_0(X)}\oX{}Z_0\delta{X}\right)/(F_0(X))^2$. Then, using the fact that $Z_0=P_0\wedge{}Q_0$ (see Lemma \ref{lemmaZ0}), and the expression (\ref{F}), one finds
\begin{equation}
d\theta_{F_0}=0. 
\label{dtheta=0}
\end{equation}
We claim that $\theta_{F_0}\in[\theta]$, where $\theta$ is the Bargmann clock introduced in Section \ref{ConfBargSection}.

\begin{pro}\label{Proxi}
The vector field $\delta_{Z_0}:X\mapsto{}Z_0X$ of $\hM_0$ descends to the quotient~$M$ defined in (\ref{M}) as a nowhere vanishing, null, vector field $\xi$, viz.,
\begin{equation}
\xi_{[X]}=D\pi(X)Z_0X.
\label{xi=pideltaZ0}
\end{equation}
\end{pro}
\begin{proof}
The derivation $\delta_{Z_0}$ preserves the constraint $\oX{}X=0$; it thus defines a vector field of $\hM_0$ which is nowhere zero because of the definition (\ref{M}). we readily check that $\delta_{Z_0}$ is invariant against dilations $X\mapsto\alpha{}X$ with $\alpha\in\bbR^*$.
The push-forward, $\xi$, of $\delta_{Z_0}$ to $M=\bbP\hM_0$ is therefore a nowhere vanishing vector field. Finally, 
(\ref{gY0}) yields $\rg_{F_0}(\xi_{[X]},\xi_{[X]})=\overline{\delta_{Z_0}X}\,\delta_{Z_0}X/F_0(X)=-\oX{}Z_0^2X/F_0(X)=0$.
\end{proof}

Let us end by proving that the vector field $\xi$ is indeed covariantly constant with respect to the Levi-Civita connection, $\nabla$, of $\rg_{F_0}$.

Applying the general formula $\nabla\theta=\half d\theta+\half{}\sfL_\xi\,\rg$, where $\theta=\rg(\xi)$, we readily find, using Equation (\ref{dtheta=0}), that $\nabla\theta_{F_0}=\half{}\sfL_\xi\,\rg_{F_0}$. Now, Equation (\ref{gY0}) helps us compute $\sfL_\xi\,\rg_{F_0}(\delta'[X],\delta''[X])=\delta_{Z_0}\left(\delta'\oX\,\delta''X/F_0(X)\right)-\overline{[\delta_{Z_0},\delta']X}\,\delta''X/F_0(X)-\delta'\oX\,[\delta_{Z_0},\delta'']X/F_0(X)=-(\delta'\oX\,\delta''X)\delta_{Z_0}F_0(X)/(F_0(X))^2=0$ since $\delta_{Z_0}F_0(X)=0$ in view of $Z_0P_0=Z_0Q_0=0$. We thus get $\nabla\theta_{F_0}=0$, hence $\nabla\xi=0$.

We have thereby proved the following proposition.

\begin{pro}\label{confBargStructure}
The triple $(M,[\rg],\xi)$ is a conformal Bargmann structure in the sense of Definition \ref{ConfBargStrDef}.
\end{pro}

\subsection{Main result: homogeneous Schr\"odinger manifolds}

Consider the triple $(\hM_\lambda,\hrg_{\lambda,\mu},\hxi)$ where $\hM_\lambda\subset\AdS_{d+3}(\sqrt{-2\lambda})$ defined by (\ref{hatM}) is the Schr\"odinger-homogeneous space (\ref{Mlambda0Ter}) endowed with the metric (\ref{hgGeneral}) and the vector field (\ref{deltaZ0Q}).  Consider next the conformal Bargmann structure $(M,[\rg],\xi)$ where $M$ defined by (\ref{M}) is the Schr\"odinger homogeneous space (\ref{MBis}) endowed via (\ref{gY0}) with the conformal class $[\rg]$ of Bargmann metrics, and where $\xi$ is the fundamental vector field  (\ref{xi=pideltaZ0}) of the group generated by $Z_0$.

\begin{thm}\label{mainThm}
The triple $(\hM_\lambda,\hrg_{\lambda,\mu},\hxi)$ is the (Poincar\'e-)Schr\"o\-dinger manifold, in the sense of Definition \ref{SchManifDef}, with conformal Bargmann boundary $(M,[\rg],\xi)$ provided $\lambda=-\half$, and $\mu=1$.
\end{thm}

\begin{proof} Our objective is thus to demonstrate that the preceding data fulfill all items of our Definition \ref{SchManifDef} of Schr\"odinger manifolds.

Let us first review some previous results expressed in local co\-ordinate systems adapted to the Schr\"odinger symmetry pervading our construction. To this purpose, and in order to make contact with the FG construction, we find it convenient to work now on the open domain where $Q\in\hM_\lambda$ admits the local form~(\ref{localQ}).
Owing to Equation~(\ref{hx}), write 
\begin{equation}
\left\{
\begin{array}{lcl}
\hx^i&=&x^i\quad(i=1,\ldots,d)\\[4pt]
\wht&=&t\\[4pt]
\hs&=&s+\lambda r^2q.
\end{array}
\right.
\label{hxhths}
\end{equation}
Now, positing
\begin{equation}
\hr=r\sqrt{-2\lambda}
\label{hr}
\end{equation}
we find 
\begin{equation}
Q
=
\frac{\sqrt{-2\lambda}}{\hr}
\pmatrix{
\hx\cr
-\half{}\hx^*\,\hx-\half\hr^2\cr
1
}\in\hM_\lambda
\label{localQbis}
\end{equation}
in view of (\ref{localQ}).
This implies that
\begin{equation}
X=
\frac{1}{r}\pmatrix{
x\cr
-\half{}x^*\,x\cr
1
}\in\hM_0
\label{localXbis}
\end{equation}
in the limit $\hr\to0$ corresponding to $\lambda\to0$, hence that a representative $[X]$ of the ray $\bbR^*X\in\bbP\hM_0$ is given, with $x\in\bbR^{d+1,1}$, by
\begin{equation}
[X]=
\pmatrix{
x\cr
-\half{}x^*\,x\cr
1
}\in{}M.
\label{local[X]}
\end{equation}

Collecting  the expressions of Section \ref{LorentzStrSubSec}, we assert that the metrics $\hrg_{\lambda,\mu}$ given by (\ref{localhgBis}), (\ref{hgGeneral}), and the vector field $\hxi$ as defined by (\ref{deltaZ0Q}), viz.,
\begin{equation}
\medbox{
\hrg_{\lambda,\mu}
=
\frac{-2\lambda}{\hr^2}\left[\sum_{i,j=1}^{d+2}\rg_{ij}\,d\hx^i\otimes{}d\hx^j+\,d\hr\otimes{}d\hr+2\lambda\mu\,\frac{d\wht\otimes{}d\wht}{\hr^2}\right]
\quad
\&
\quad
\hxi=\frac{\partial}{\partial\hs}
}
\label{localhgxi}
\end{equation}
constitute a family of Lorentz metrics (for $\lambda<0$), while $\hxi$ is a nowhere zero null Killing vector field.

The crux of the matter is that $\hr=r\sqrt{-2\lambda}$ defined by (\ref{hr}) is definitely (see~(\ref{g+}) and (\ref{localhgxi})) our \textit{defining function} for $M$, the conformal boundary of $\hM_\lambda$ (coordinatized as in (\ref{local[X]})).
Moreover Equations (\ref{hxhths}) and (\ref{hr}) entail that, locally, $\partial/\partial{\hs}=\partial/\partial{s}$, a relationship which is consistent with the limit $\hr\to0$. The vector field $\hxi=\partial/\partial{\hs}$ of $\hM_\lambda$ therefore goes smoothly over to $M=\partial\hM_\lambda$ as the vector field $\xi=\partial/\partial{s}$. This justifies, in local terms, item~\ref{one} of Definition \ref{SchManifDef}, the latter being globally accounted for by Proposition \ref{Proxi}.

Furthermore, straightforward computation using (\ref{localhgxi}) shows that
\begin{equation}
\hrg^{\lambda,\mu}=
\frac{-\hr^2}{2\lambda}\left[\sum_{i,j=1}^{d+2}\rg^{ij}\,\frac{\partial}{\partial{\hx^i}}\otimes\frac{\partial}{\partial{\hx^j}}+\frac{\partial}{\partial{\hr}}\otimes\frac{\partial}{\partial{\hr}}\right]
+
\mu\,\frac{\partial}{\partial{\hs}}\otimes\frac{\partial}{\partial{\hs}}.
\label{ginverse}
\end{equation}
This readily yields $\hrg^{\lambda,\mu}\big\vert_{T^*M}=\mu\,\xi\otimes\xi$ in the limit $\hr\to0$, insuring that item \ref{two} of Definition \ref{SchManifDef} holds true prior to imposing the normalization condition $\mu=1$.

Using then the form of the Poincar\'e metric $\rg^+$ of Definition~\ref{SchManifDef}, we easily deduce from (\ref{localhgxi}) that $\rg^+_{\lambda}=\hrg_{\lambda}$ as given by Equation (\ref{localhgBis}).
This proves that it is only the conformal class, $[\rg]$, of the metric $\rg=\lim_{\lambda\to0}(r^2\,\rg_\lambda^+)$ that goes over to the boundary $M=\bbP\hM_0$ of $\hM_\lambda$.
To sum up, we find that the triple $(M,[\rg],\xi)$ where
\begin{equation}
\medbox{
\rg=\sum_{i=1}^d dx^i\otimes{}dx^i+2dt\odot{}ds
\Bigg(=\sum_{i,j=1}^{d+2}\rg_{ij}\,dx^i\otimes{}dx^j\Bigg)
\quad
\&
\quad
\xi=\frac{\partial}{\partial s}
}
\label{BargHomSchr}
\end{equation}
is a representative of our Schr\"odinger-homogeneous conformal Barg\-mann structure (see (\ref{FlatBargmannStructure})) expres\-sed in the adapted local coordinate system provided by (\ref{local[X]}). 
Direct computation moreover shows that 
\begin{equation}
\Ric(\rg^+_\lambda)+(d+2)\rg^+_\lambda=\frac{(d+2)(1+2\lambda)}{2\lambda}\rg^+_\lambda
\label{EinsteinEqs}
\end{equation}
which enables us to conclude that $\rg^+_\lambda$ is indeed a Poincar\'e metric on $\hM_\lambda$, consistently with Definition \ref{FG}, if the right-hand side of Equation (\ref{EinsteinEqs}) vanishes, i.e., if $\lambda=-\half$. Item \ref{three} of the definition \ref{SchManifDef} of Schr\"odinger manifolds is therefore fulfilled. 

The proof of Theorem \ref{mainThm} is complete.
\end{proof}

We refer to Figure \ref{Fig1} for a graphical representation of our construction.


\begin{figure}[h]
\begin{center}
\begin{tikzpicture} [scale=1.055]
\draw[line width=2pt] (0,0) circle (5cm);
\node [rectangle,rounded corners,draw,name=RP] at (-5,4.5) {$\bbP^{d+3}(\bbR)$};
\draw [->,decorate,decoration={snake,amplitude=.7mm,segment length=7mm,post length=1mm}] (RP) -- (-3.7,3.5);
\filldraw[black!40,line width=1.cm] (0,0) circle (4cm);
\node[rectangle,rounded corners,draw,fill=black!40,name=AdS] at (6,4.5) {$\AdS_{d+3}(\sqrt{-2\lambda})$};
\draw [->,decorate,decoration={snake,amplitude=.7mm,segment length=7mm,post length=1mm}] (AdS) -- (3.5,2.5);
\filldraw[black!10] (0,0) circle (3.8cm);
\node[rectangle,rounded corners,draw,fill=black!10,name=hM] at (-5,-4.5) {$(\hM_\lambda,\hrg_{\lambda,\mu},\hxi)$};
\draw [->,decorate,decoration={snake,amplitude=.7mm,segment length=7mm,post length=1mm}] (hM) -- (-2.3,-2.3);
\draw[very thick,->] (0,1.06) -- (0,3.8);
\node at (.47,1.2) {\footnotesize $\hr=0$};
\node at (.2,3.5) {\footnotesize $\hr$};
\draw[line width=4pt,black!40,fill=white]  (0,0)  circle (.91cm);
\node[name=Ein] at (0,-.2) {$\Ein_{d+1,1}$};
\draw [->,decorate,decoration={snake,amplitude=.4mm,segment
  length=4mm,post length=1mm}] (-.2,0) -- (.59,.59);
\draw[line width=3pt,rotate=-90,black]  (1,.06) arc (3.5:356.5:1cm);
\node[rectangle,rounded corners,draw,name=M] at (6,-4.5) {$(M=\partial \hM_\lambda,[\rg],\xi)$};
\draw [->,decorate,decoration={snake,amplitude=.7mm,segment length=7mm,post length=1mm}] (M) -- (.75,-.75);
\node[rectangle,draw,name=pt] at (-2,-6) {$S^1\times\{\mathsf{pt}\}$};
\draw [thick,->]
(pt) -- (0,-1.1);
\begin{scope}[>=latex]
\draw[very thick,->,rotate=-25] (1.05,0) -- (1.05,1.2) node[very near end,sloped,below,rotate=-90] {$\quad\xi$};
\draw[very thick,->,rotate=-45] (-1.05,0) -- (-1.05,-1.2) node[very near end,sloped,below,rotate=90] {$\quad\xi$};
\end{scope}
\coordinate (C) at (1.6,1.6);
\begin{scope}[rotate=20]
\begin{scope}[>=latex]
\draw[very thick,->]  ($(C) + (-.5,-.5)$) -- ($(C) + (1,1)$) node[very near end,sloped,below,rotate=-45] {$\quad\hxi$};
\end{scope}
\draw  ($(C) + (.5,-.5)$) -- ($(C) + (-.5,.5)$);
\draw ($(C) +  (.5,-.5)$) -- ($(C) + (-.5,.5)$);
\shadedraw [shading angle = -90] ($(C) + (0,0.5)$) ellipse (.5cm and .1cm);
\shadedraw  ($(C) + (0,-0.5)$) ellipse (.5cm and .1cm);
\end{scope}
\coordinate (C) at (-2.3,-.3);
\begin{scope}[rotate=175]
\begin{scope}[>=latex]
\draw[very thick,->]  ($(C) + (-.5,-.5)$) -- ($(C) + (1,1)$) node[very near end,sloped,below,rotate=-45] {$\quad\hxi$};
\end{scope}
\draw  ($(C) + (.5,-.5)$) -- ($(C) + (-.5,.5)$);
\draw ($(C) +  (.5,-.5)$) -- ($(C) + (-.5,.5)$);
\shadedraw [shading angle = -90] ($(C) + (0,0.5)$) ellipse (.5cm and .1cm);
\shadedraw  ($(C) + (0,-0.5)$) ellipse (.5cm and .1cm);
\end{scope}
\end{tikzpicture}
\end{center}
\caption{Schr\"odinger manifold \& conformal Bargmann boundary}
\label{Fig1}
\end{figure}
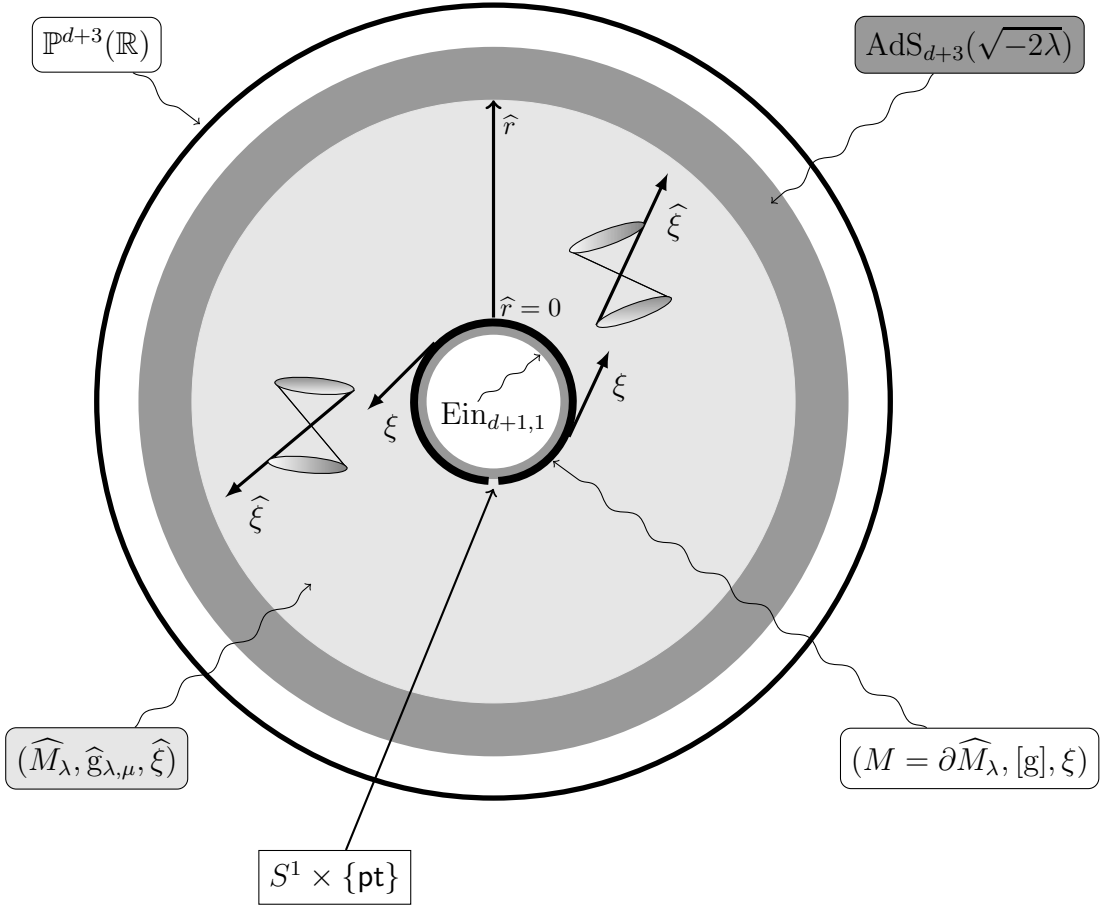

\begin{rmk}
{\rm
We duly recover the metric (\ref{BSmetric}) of the ``AdS/CFT'' correspondence from the metric given by (\ref{localhgxi}) by imposing the special values $\lambda=-\half$, and $\mu=1$ yielded by Theorem \ref{mainThm}.
}
\end{rmk}

\begin{rmk}
{\rm
It can be checked that there holds, in full generality,
\begin{equation}
\Ric(\hrg_{\lambda,\mu})-\frac{d+2}{2\lambda}\,\hrg_{\lambda,\mu}=-\mu\frac{d+4}{2\lambda}\,\htheta\otimes\htheta.
\label{TracefreeEinstein}
\end{equation}
These equations are interpreted as Einstein's equations $\widehat{\Ric}-\half\widehat{R}\,\hrg+\Lambda\,\hrg=\rT$ with a cosmological constant $\Lambda=(d+1)(d+2)/(4\lambda)$, and sources given in terms of the ``null fluid'' stress-energy-momentum tensor $\rT=-\mu(d+4)/(2\lambda)\,\htheta\otimes\htheta$. See also~\cite{DHH}.
}
\end{rmk}

\begin{rmk}
{\rm
We learn from Equation (\ref{thetadtheta=0}) that 
the distri\-bution $\ker(\htheta)$ is actual\-ly integra\-ble. This is the very condition found in \cite{JN} to achieve a null-Killing dimensional reduction. Our Schr\"odinger-homogeneous manifolds $(\hM_{\lambda},\hrg_{\lambda,\mu},\hxi)$ thus provide examples of those manifolds considered by Julia and Nicolai~\cite{JN}.
}
\end{rmk}

\section{Conclusion}\label{ConclusionSection}


This article has been triggered by the seemingly contradictory emergence of non-relativistic Schr\"odinger ``isometries'' within the framework of an \textit{a priori} relativistic AdS/CFT correspondence (in the case where the dynamical exponent is $z=2$).

A closer look at the literature referred to in the introduction made it clear that the metric (\ref{BSmetric}) appearing in the physics of non-relativistic holography should be related to the structure of what has been called a Bargmann extension of non\-relativistic spacetime; see Section \ref{ConfBargSection} which also offers a general definition of the Schr\"odinger group. This hint has been first investigated in \cite{DHH}. Our task, here, was thus to put this observation on more global geometrical grounds.

From this vantage point, we have chosen to specialize the construction of a Poincar\'e metric, due to Fefferman and Graham, to the case where conformal infinity is moreover endowed with a conformal Bargmann structure governed by a null, parallel, vector field. This has led us to our definition \ref{SchManifSection} of Schr\"odinger manifolds.

Let us insist that the general proof of the existence and uniqueness (in suitable  dimensions) of Schr\"odinger prolongations of Bargmann manifold structures has not been envisaged here, being clearly beyond the scope of this article. This will be deferred to sub\-sequent work. 

Nevertheless, the purpose of this article is to supply explicit examples of such Schr\"o\-dinger manifolds that would help us understand the origin of the above-mentioned metric, with Schr\"odinger isometries, in a non-relativistic avatar of the AdS/CFT correspondence.
Accordingly, we have found it useful to characterize, in the ``flat'' case, the Schr\"odinger group, $\Sch(d+1,1)$, as the stabilizer within $\rO(d+2,2)$ of a distinguished nilpotent element, $Z_0$, of the Lie algebra, $\ro(d+2,2)$. Our construction interestingly confers, as awaited and in a clear-cut fashion, a non-relativistic status to the Schr\"odinger group within a purely relativistic framework.

Our main result, namely Theorem \ref{mainThm}, provides us with examples of Schr\"odinger manifolds, $(\hM_\lambda,\hrg_{\lambda,\mu},\hxi)$; the canonical one is fixed by the normalization conditions $\lambda=-\half$, and $\mu=1$. Note that $\hM_\lambda$ is actually a homogeneous space of the Schr\"odinger group $\Sch(d+1,1)$, and, besides, an open submanifold of $\AdS_{d+3}(\sqrt{-2\lambda})$. In a appropriate coordinate system on $\hM_{-\half}$, the metric $\hrg_{-\half,1}$ matches exactly the Balasubramanian-McGreevy and Son metric (\ref{BSmetric}). See also \cite{SY2} for a local approach in terms of a non-reductive homogeneous space of the Schr\"odinger group. 

Let us stress that it finally appears that the Schr\"odinger group $\Sch(d+1,1)$ is, as expected, the maximal group of isometries of our Schr\"odinger manifolds. This definitely firms up the claims of \cite{Bala,Son}. 

There remains, however, to understand, in completely general terms, the relation\-ship between the Schr\"odinger group defined as the group of automorphisms of a conformal Bargmann structure and the group of automorphisms of an as\-sociated (Poincar\'e-)Schr\"odinger structure. This program for future work should indeed take advantage of a key result of Anderson \cite{And} about the isometric extensions of the automorphisms of conformal infinity of a conformally compact Einstein manifold.

From another perspective, it would be worthwhile considering our construction of Poincar\'e-Schr\"odinger metrics for the canonical circle-bundle of a CR manifold (see, e.g., \cite{DT} for a general reference on CR geometry) endowed with its Fefferman metric, and a null nowhere vanishing Killing vector field, given by the genera\-tor of the $S^1$-action \cite{Fef,Lee,Gra,BD}.

\goodbreak

We finally expect that the definition of Schr\"odinger manifolds put forward in this article, and the explicit examples that have been worked out, will foster new research in the very attractive domain of non-relativistic AdS/CFT correspondence.

\paragraph{Acknowledgements:}
{We would like to express our gratitude to P. Horv\'athy for his interest, stimulating encouragements, and also for his help during the first stage of the preparation of the manuscript. Thanks are also due to X. Bekaert for most enlight\-ening discussions, and to J.~Hartong and S.~Detournay for useful cor\-respondence. We furthermore acknowledge the referees' valuable suggestions.}



\end{document}